\documentclass[a4paper,USenglish,cleveref, autoref, thm-restate]{lipics-v2021}
\nolinenumbers

\usepackage{graphicx}
\usepackage{multicol}

\newcommand{\para}[1]{\subparagraph*{#1}}

\newcommand{\N}{\mathbb{N}}

\newcommand{\eps}{\varepsilon}

\newcommand{\floor}[1]{\left\lfloor{#1}\right\rfloor}

\newcommand{\Otild}{\tilde{O}}
\newcommand{\TSHD}{\text{Hamming Distance Oracle}}
\newcommand{\HD}{\mathsf{HD}}

\newcommand{\D}{\mathsf{D}}

\newtheorem{problem}{Problem}

\title{Hamming Distance Oracleirpin Completion Distance Lower Bound}

\author{Itai Boneh}{Reichman University and University of Haifa, Israel}{itai.bone@biu.ac.il}{https://orcid.org/0009-0007-8895-4069}{supported by Israel Science Foundation grant 810/21.}

 \author{Dvir Fried}{Bar Ilan University, Israel}{friedvir1@gmail.com}{https://orcid.org/0000-0003-1859-8082}{supported by ISF grant no. 1926/19, by a BSF grant 2018364, and by an ERC grant MPM under the EU's Horizon 2020 Research and Innovation Programme (grant no. 683064).}

  \author{Shay Golan}{Reichman University and University of Haifa, Israel}{golansh1@biu.ac.il}{https://orcid.org/0000-0001-8357-2802}{supported by Israel Science Foundation grant 810/21.}

    \author{Matan Kraus}{Bar Ilan Univesity, Israel}{matan3@gmail.com}{https://orcid.org/0000-0002-2989-1113}{supported by the ISF grant no. 1926/19, by the BSF grant 2018364, and by the ERC grant MPM under the EU's Horizon 2020 Research and Innovation Programme (grant no. 683064).}

\authorrunning{Boneh et al.} 


\newtheorem{fact}{Fact}

\crefname{fact}{Fact}{Facts}
\crefname{problem}{Problem}{Problems}

\hideLIPIcs  

\title{Hamming Distance Oracle}

\begin{document}

\maketitle
\begin{abstract}
In this paper, we present and study the \emph{Hamming distance oracle problem}.
In this problem, the task is to preprocess two strings $S$ and $T$ of lengths $n$ and $m$, respectively, to obtain a data-structure that is able to answer queries regarding the Hamming distance between a substring of $S$ and a substring of $T$. 

For a constant size alphabet strings, we show that for every $x\le nm$ there is a data structure with $\Otild(nm/x)$ preprocess time and $O(x)$ query time.
We also provide a combinatorial conditional lower bound, showing that for every $\eps > 0$ and $x \le nm$ there is no data structure with query time $O(x)$ and preprocess time $O((\frac{nm}{x})^{1-\eps})$ unless combinatorial fast matrix multiplication is possible.

For strings over general alphabet, we present a data structure with $\Otild(nm/\sqrt{x})$ preprocess time and $O(x)$ query time for every $x \le nm$.

\end{abstract}

\clearpage
\setcounter{page}{1} 

\section{Introduction}\label{sec:intro}

Given two strings $S,T$ of the same length $n$, the hamming distance $\HD(S,T)$ is the number of mismatches between $S$ and $T$.
Hamming distance is arguably the most basic and common measure of similarity between strings.
The computational task of finding the Hamming distance of two given strings is trivial - it can be done in $O(n)$ time and nothing faster is possible.

In recent years, many classical string problems have been considered in the substring queries model. 
In this model, we are interested in preprocessing a string to obtain a data structure capable of efficiently answering queries regarding its substrings.
A few examples of results in the substring queries model include finding the longest increasing subsequence of a substring~\cite{T2008}, the longest common substring of two substrings~\cite{ACPR2020}, finding the period of a substring ~\cite{KRRW2012}, and applying pattern matching~\cite{KKFL2014,KRRW2014} and approximate pattern matching between substrings~\cite{CKW2020}.

In this paper, we present the natural problem of constructing a \textit{Hamming distance oracle}.

\begin{problem}[Hamming Distance Oracle]\label{prob:HDO}
Given two strings $S,T\in\Sigma^*$ of lengths $n$ and $m$, respectively, preprocess a data-structure that supports the following query:
For a length $\ell$ and two indices $i\in[n-\ell+1]$ and $j\in[m-\ell+1]$, compute $\HD(S[i..i+\ell-1],T[j..j+\ell-1])$.
The quality of the oracle is measured by its preprocessing time, and by its query time.
\end{problem}

For constant size alphabet strings, we show that for every $x\le nm$ there is a data structure with $\Otild(nm/x)$ preprocess time and $O(x)$ query time.
We also provide a combinatorial conditional lower bound, showing that our trade-off is optimal for combinatorial algorithms.

For strings over general alphabet, we present a data structure with $\Otild(nm/\sqrt{x})$ preprocess time and $O(x)$ query time for every $x \le nm$. See \cref{fig:complex} for illustration.

\begin{figure}[ht!]
  \begin{center}
 \includegraphics[scale=0.5]{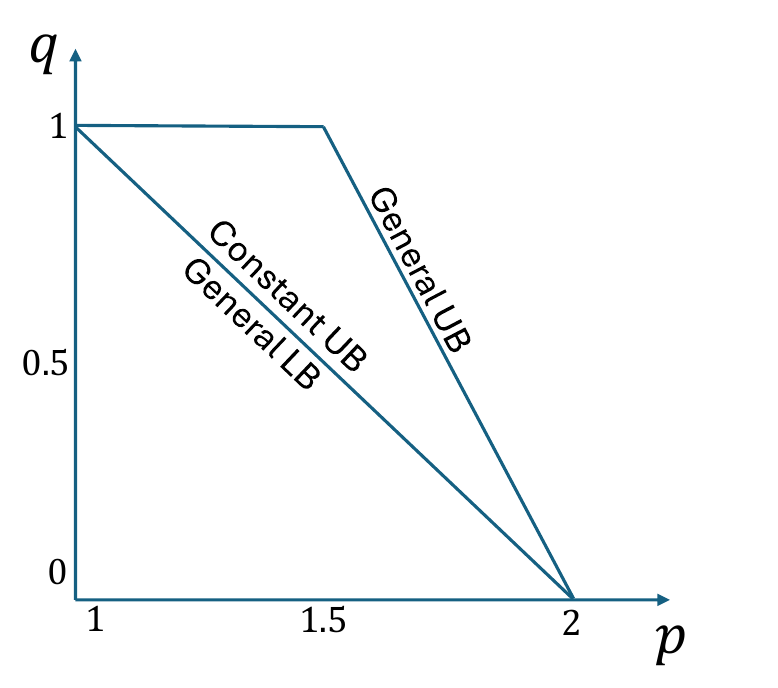} 
  \caption{A summary of our results for $n=m$. 
  the $p$-axis corresponds to the exponent of the preprocess time and the $q$-axis corresponds to the exponent of the query time.
  For example, we have a general upper bound of $\Otild(n^{1.75})$ preprocess time and $O(\sqrt{n})$ query time.
  Note that the lower bound is combinatorial.
  \label{fig:complex}}
   \end{center}
   \vspace*{-2mm}
\end{figure}

\para{Related work.} Charalampopoulos et. al ~\cite{CGMW21} considered the problem of constructing an \textit{Edit Distance Oracle}.
That is, preprocess two input strings $S$ and $T$ to obtain a data structure that is able to compute the edit distance of two given substrings of $S$ and $T$.
Charalampopoulos et al. ~\cite{CGMW21} provided an optimal (up to sub-polynomial factors) data structure with $O(N^{1+o(1)})$ preprocess time and $O(\mathbf{polylog}(n))$ query time with $N=|S| \cdot |T|$.
It is quite surprising that the problem of constructing a Hamming distance oracle has not been considered, as Hamming distance is arguably more fundamental than edit distance.

\para{Organization} in \cref{sec:prelim} we present basic notations and definitions. 
In \cref{sec:upper_bound} we prove the stated upper bound (see \cref{thm:upper_bound}). In \cref{sec:lower_bound} we prove the stated lower bounds (see \cref{thm:lb}).


\section{Preliminaries}\label{sec:prelim}

For $i,j\in\N$ we denote $[i..j] =\{k \in \N \mid i\le k \le j \}$ and $[i]=[1..i]$.

A string $S$ of length $|S|=n$ over an alphabet $\Sigma$ is a sequence of characters $S=S[1]S[2]\ldots S[n]$.
For $i,j\in[n]$, we call $S[i..j]=S[i]S[i+1]\ldots S[j]$ a \emph{substring} of $S$. If $i=1$, $S[i..j]$ is a prefix of $S$, and if $j=|S|$, $S[i..j]$ is a suffix of $S$.
Let $S$ and $T$ be two strings over an alphabet $\Sigma$.
$S\cdot T$ is the concatenation of $S$ and $T$.

For two strings $S$ and $T$ of the same length $n$, the Hamming distance~\cite{Hamming50} of $S$ and $T$ is defined as 
$\HD(S,T)=|\{ i\in[n]\mid S[i]\neq T[i]\}|$.



\section{Hamming Distance Oracle}\label{sec:upper_bound}
Here we introduce our upper bound for the $\TSHD$ problem.

We first reduce \cref{prob:HDO} to the problem of computing the Hamming distance of two \emph{suffixes}. 
In order to do so, we slightly abuse the notation and define the Hamming distance between two strings of \emph{different} lengths as follows:
Let $S$ and $T$ be two strings of lengths $n$ and $m$ respectively, we define the Hamming distance between $S$ and $T$ to be the Hamming distance between their prefixes of length $\min\{n,m\}$. Formally, $\HD(S,T)=\HD(S[1..\min\{n,m\}],T[1..\min\{n,m\}])$.

\begin{problem}[Suffixes Hamming Distance Oracle]\label{prob:SHDO}
Given two strings $S,T\in\Sigma^*$ of lengths $n$ and $m$, respectively, preprocess a data-structure that supports the following query:
For two indices $i\in[n]$ and $j\in[m]$, compute $\HD(S[i..n],T[j..m])$.
The complexity of the oracle is measured by its preprocessing time, its query time, and its space usage.
\end{problem}

Due to the following fact, given an oracle for \cref{prob:SHDO}, one can answer any query of \cref{prob:HDO} by two queries to the oracle. 

\begin{fact}\label{fact:HD_diff}
    For every $i\in  [n]$, $j \in [m]$ and $\ell \in [\min(n-i,m-j)]$, we have $\HD(S[i..i+\ell-1],T[j..j+\ell-1]) = \HD(S[i..n],T[j..m])-\HD(S[i+\ell..n],T[j+\ell..m])$.
\end{fact}
From now on we will focus on introducing an oracle for \cref{prob:SHDO}.
The running time of our oracle (see \cref{thm:upper_bound}) depends on the complexity of the Text-to-Pattern Hamming distance problem.
Let $T_\HD(n,m,\Sigma)$ be the time complexity of computing the Text-to-Pattern Hamming Distance where $n$ is the length of the text, $m$ is the length of the pattern and both strings are over an alphabet $\Sigma$.
Notice that for any alphabet $\Sigma$ we have $T_{\HD}(n,m,\Sigma)=O(|\Sigma|\cdot n\log m)$ using FFT (by Fischer and Paterson~\cite{FP74}).
For general alphabet $\Sigma$ we have $T_{\HD}(n,m,\Sigma)=O(n\sqrt m)$ randomized (by Chan et al.~\cite{CJVWX23}) or $T_{\HD}(n,m,\Sigma)=O(n\sqrt {m\log\log m})$ deterministic (by Jin and Xu~\cite{JX24}).

\begin{theorem}\label{thm:upper_bound}
    Fix $x\geq 1$.
    Given two strings $S, T$ over an alphabet $\Sigma$, such that $|S|=n$, $|T|=m$ and $m\leq n$.
    There exists a data structure for \cref{prob:SHDO} with a preprocessing time of $O(\frac{n}{x}\cdot T_{\HD}(m,x,\Sigma))$, and a query time of $O(\min(m,x))$.

\end{theorem}
\begin{proof}
We first present a simple dynamic programming algorithm, which prove the theorem for $x=1$.
We then show how we can compute only portion of the dynamic programming table, and then bound the time the data structure needs for answering a query.

Let $\D$ be a matrix of size $n\times m$ such that $\D[i,j]=\HD(S[i..n],T[j..m])$, i.e. the hamming distance of the $i$th suffix of $S$ and the $j$th suffix of $T$.
Then, the matrix $D$ satisfies the following recursion:
\begin{equation}\label{eq:still_basic_recursion}
	\D[i,j] =
	\begin{cases}
		\HD(S[i], T[j]) & \text{ if } i=n \vee j=m,\\
		\HD(S[i],T[j]) + \D[i+1,j+1] &\text{otherwise}
	\end{cases}
\end{equation}
Notice that by \cref{eq:still_basic_recursion} one can compute each value $D[i,j]$ in $O(1)$ time (assuming a right order of computation).
Thus, filling the dynamic programming table takes $O(nm)$ time in total.
Due to \cref{fact:HD_diff}, the table $D$ completes the proof for $x=O(1)$.



We proceed to describe our construction for any $x \ge 1$.
\para{Preprocessing.}
The data structure maintains a small portion of the table $\D$:
For each $i\in[1..\floor{\frac{n}{x}}]$,
the algorithm computes and stores the $i\cdot x$'s row, i.e. for each $(i,j) \in[1..\floor{\frac{n}{x}}] \times [m]$ the data structure computes the cell $\D[i\cdot x,j]$.
Notice that (by \cref{fact:HD_diff}), the value of any cell $D[i,j]$ is the sum of the Hamming distance  $\HD(S[i..i+\ell-1],T[j..j+\ell-1])$ and $D[i+\ell,j+\ell]$ (if the indices of the cell are not in the table, we consider the cell value as $0$). 
In particular,
\begin{align*}
\D[i\cdot x,j] &=\HD(S[i\cdot x..n],T[j..m])\\
&=\HD(S[i\cdot x..(i+1)x-1],T[j..j+x-1]) + \HD(S[(i+1)x..n],T[j+x..m])\\
&=\HD(S[i\cdot x..(i+1)x-1],T[j..j+x-1]) + \D[(i+1)\cdot x,j+x].
\end{align*}
Therefore, by computing the rows in decreasing order (and using cells from rows already computed), the data structure is able to compute each row with a single invocation of a text-to-pattern Hamming distance algorithm with text $T$ and pattern $S[i\cdot x .. (i+1)x -1]$.
In every such invocation, the text's size is $|T|=m$, the pattern size is $x$, and the strings are over the same alphabet and therefore the running time per computed row is $T_{\HD}(m,x,\Sigma)$.
Thus, the preprocessing time is $O(\frac nx\cdot T_\HD(m,x,\Sigma))$ as required.
\para{Query.}
Without loss of generality, we focus on the case $x\leq m$, since otherwise one can answer a query in $O(m)\subseteq O(x)$ time na\"ively by comparing up to $m$ pairs of characters, even without preprocessing.
Let $a$ be the minimal integer multiply of $x$ larger or equal to $i$.
The algorithm first computes $\HD(S[i..a-1],T[j..a-1])$ na\"ively by comparing pairs of corresponding characters.
Then the algorithm returns $\HD(S[i..a-1],T[j..a-1])+D[a,j+a-i]$.
Since $a$ is an integer multiple of $x$, we have that $D[a,j+a-i]$ is accessible in $O(1)$ time after the preprocessing (in the case where $(a,j+a-i)$ are outside the boundaries of the table, then $D[a,j+a-i]=0$ and is also accessible in $O(1)$ time).
Therefore the running time of the query is $O(a-i)=O(x)$.
The correctness follows immediately from \cref{fact:HD_diff} (see \cref{fig:query}).\qedhere

\begin{figure}[ht!]
  \begin{center}
 \includegraphics[scale=0.5]{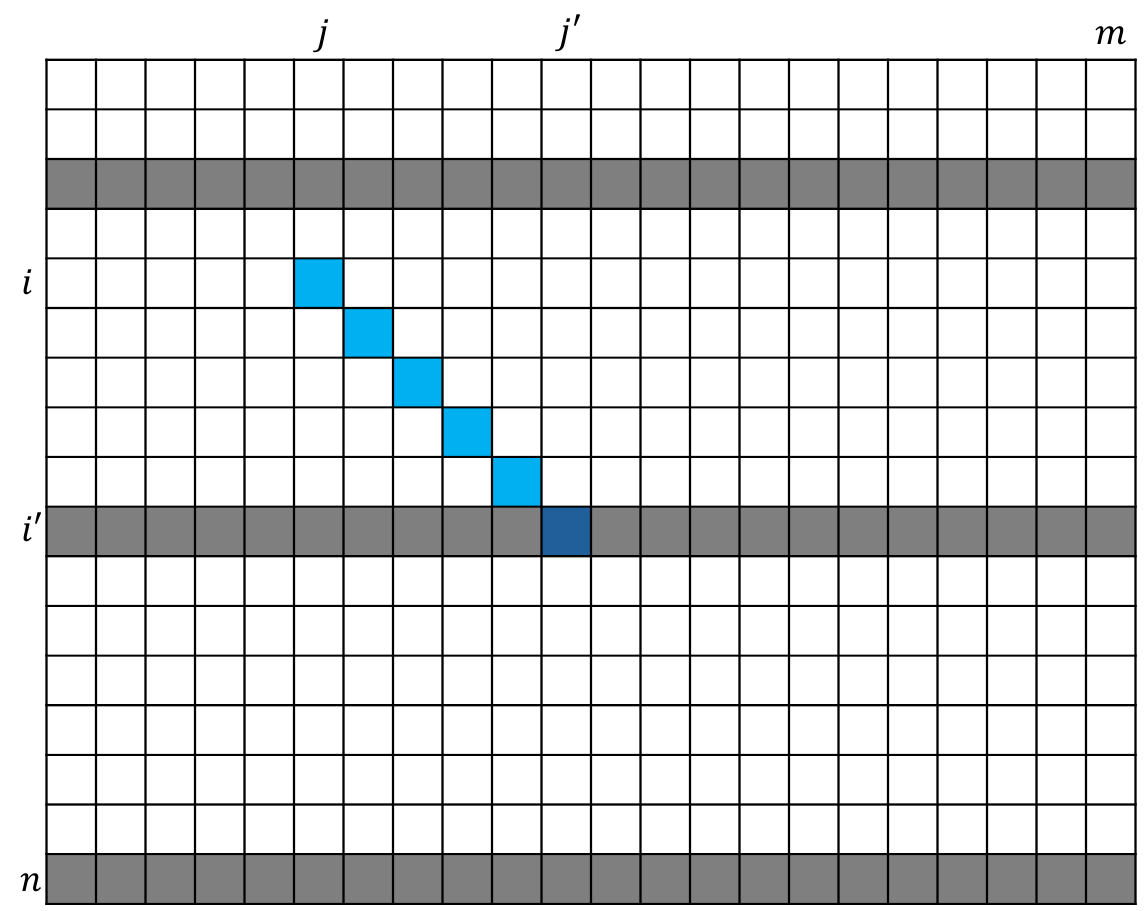} 
  \caption{An example of a query.
  The grey rows are computed in the prerocess.
  In order to compute $\HD(S[i..n],T[j..m])$, it is enough to compute $\HD(S[i..i'-1],T[j..j'-1])$ na\"ively and add $\D(i',j')$.
  \label{fig:query}}
   \end{center}
   \vspace*{-2mm}
\end{figure}


\end{proof}

\section{Lower Bound for the binary case of $\TSHD$}\label{sec:lower_bound}
In this section we establish the hardness of \cref{prob:HDO}, at least for combinatorial algorithms.
The lower bound is based on the known conjecture that combinatorial boolean matrix multiplication cannot be (polynomialy) faster than the na\"ive algorithm. 

\begin{conjecture}[Combinatorial Matrix Multiplication, see~\cite{GU18}]\label{conj:cmm}
For any $\alpha,\beta,\gamma,\varepsilon> 0$, there is no
combinatorial algorithm for multiplying an $n^\alpha \times n^\beta$ matrix with an $n^\beta \times n^\gamma$ matrix in time
$O((n^{\alpha + \beta + \gamma})^{(1-\varepsilon)})$ \footnote{~\cite{GU18} stated the running time as $O(n^{\alpha +\beta+\gamma - \varepsilon})$. It can be easily shown that our statement is equivalent.}.
\end{conjecture}

Based on \cref{conj:cmm} we establish the hardness of \cref{prob:HDO} even for binary alphabet as follows.


\begin{theorem}[Lower bound]\label{thm:lb}
Let $O$ be a combinatorial oracle for \cref{prob:HDO} for strings over an alphabet $\Sigma$ with $|\Sigma|\ge 2$.
Let $p(n,m)$ and $q(n,m)$ be the preprocessing and query time of $O$, respectively.
Then there is no $\eps>0$ such that $p(n,m)\cdot q(n,m)=O((nm)^{1-\eps})$, unless \cref{conj:cmm} is false.
\end{theorem}



\begin{proof}
    As mentioned in \cref{sec:intro}, our reduction is based on idea attributed to Ely Porat and Piotr Indyk~\cite{Clifford09} that was later generalized by Gawrychowski and Uznanski~\cite{GU18}.

    Assume by contradiction that there exists some $\eps>0$ such that $p(n,m)\cdot q(n,m)=O((nm)^{1-\eps})$.
    Let $x=q(n,m)\cdot (nm)^{\eps/2}$ and note that $q(n,m)=\frac{x}{(nm)^{\eps/2}}$.
%
    %
    Let $A$ and $B$ be two boolean matrices of sizes $\frac nx\times x$ and $x\times \frac mx$, respectively.
    The idea is to create a string $S$ representing the rows of $A$ and a string $T$ representing the columns of $B$.
    We transform each row of $A$ and each column of $B$ such that the Hamming distance between the row's string and the column's string indicates the boolean product of the corresponding row and column.
    We first show a construction using ternary alphabet, and then explain how one can convert it into binary alphabet.

    \para{Encoding.}
    We define encoding of values into characters, which is as follows:
    \begin{itemize}
        \item each $1$ is encoded by the character $'1'$.
        \item each $0$ from $A$ is encoded by the character $'x'$.
        \item each $0$ from $B$ is encoded by the character $'y'$.
    \end{itemize}
    For each $i\in [\frac nx]$ the algorithm encodes the $i$th row of $A$ as $A_{i}$, and for each $j\in [\frac mx]$ the algorithm encodes the $j$th column of $B$ as $B_{j}$.
    Then $S$ is the concatenation of $A$'s rows (i.e. $S=A_1\cdot A_2\cdot...A_{n/x}$), and $T$ is the concatenation of $B$'s columns ($T=B_1\cdot B_2\cdot...B_{m/x}$).
    The resulted strings $S$ and $T$ has length of $n$ and $m$ (respectively), and then the algorithm creates the Hamming Distance Oracle data structure  $O$ of $S$ and $T$.

    In order to explain how the algorithm computes the boolean product $AB$ using queries on $S$ and $T$, we introduce the following observation.

    \begin{lemma}\label{obs:dah}
        $(AB)_{ij}=1 \iff \HD(A_{i},B_{j})<x$
    \end{lemma}
    \begin{proof}
        $(AB)_{ij}=1$ if and only if there exists some $k \in [x]$ such that $A[i,k]=1=B[k,j]$.
        That means that both $A[i,k]$ and $B[k,j]$ were encoded by the character $'1'$, and $\HD(A[i,k],B[k,j])=0$.
        This implies that for $A_{i}$ and $B_{j}$ there is a match in at least one index, and therefore $\HD(A_{i},B_{j})<x$.
        All the transitions are bidirectional, so the proof is complete.
    \end{proof}

    To prove the hardness for binary alphabet
    one can find binary encoding for each symbol in $\{1,x,y\}$ that preserves the mismatches: for each $a\neq b \in \{1,x,y\}$ we can encode $a$ and $b$ such that $\HD(E(a),E(b))=k$ for some positive integer $k=O(1)$
    , and then check if the distance between $A_i$ and $B_j$ is strictly less than $k\cdot x$.
    An example for such encoding is $'011','101','110'$.
    For this encoding, it holds that the Hamming distance of any two different words is $2$.
    
    Due to \cref{obs:dah} and the discussion above, the algorithm would query the data structure for each entry of $AB$ if $\HD(A_{i},B_{j})<2x$ and fill in the matrix $AB$ accordingly.

    The correctness comes directly from \cref{obs:dah}.
    Thus, using the oracle $O$, one can compute the boolean matrix product $AB$ in $O(p(n,m)+\frac{nm}{x^2}q(n,m))$ time.
    
    Recall that $p(n,m)\cdot q(n,m)=O((nm)^{1-\eps})$ and
    $q(n,m)=\frac{x}{(nm)^{\eps/2}}$.
    Thus, we get that     \[ p(n,m)=O((nm)^{1-\eps}/q(n,m))=O((nm)^{1-\eps/2}/x)=O\left(\left(\frac{nm}{x}\right)^{1-\eps/2}\right).\]

    Moreover,     
    \[\frac{nm}{x^2}q(n,m)=
    \frac{nm}{x^2}\frac x{(nm)^{\eps/2}}=
    \frac{nm}{x}\frac 1{(nm)^{\eps/2}}=
    \frac{(nm)^{1-\eps/2}}{x}=O\left(\left(\frac{nm}{x}\right)^{1-\eps/2}\right).\]
    Thus, the time for computing the product $AB$ using the oracle $O$ is $O((\frac {nm}{x})^{1-\eps/2})$, which contradicts \cref{conj:cmm} that suggests that no algorithm can compute $AB$ in $O((\frac {nm}{x})^{1-\eps'})$ time.
\end{proof}


\bibliography{bib}

\end{document}